\documentclass[11pt]{article}
\usepackage[margin=0.92in,headheight=14pt]{geometry}
\usepackage[T1]{fontenc}
\usepackage{lmodern,microtype}
\usepackage{amsmath,amssymb,amsthm,mathtools}
\usepackage{booktabs,tabularx,array}
\usepackage{graphicx,xcolor,listings}
\usepackage{enumitem,placeins}
\usepackage{caption}
\usepackage[numbers,sort&compress]{natbib}
\usepackage{hyperref}
\usepackage[nameinlink,capitalise,noabbrev]{cleveref}
\definecolor{navy}{HTML}{18364C}
\definecolor{teal}{HTML}{167D8D}
\definecolor{codebg}{HTML}{F5F7F9}
\hypersetup{colorlinks=true,linkcolor=navy,citecolor=teal,urlcolor=teal,
 pdftitle={High-Rate Quantum Codes with Proven Distance and Low-Weight Measurements},
 pdfauthor={Kishor Bharti, Tobias Haug, Runzhou Tao, Kevin Ye}}
\setlist{itemsep=2pt,topsep=4pt,leftmargin=1.5em}
\newtheorem{theorem}{Theorem}[section]
\newtheorem{lemma}[theorem]{Lemma}
\newtheorem{proposition}[theorem]{Proposition}
\newtheorem{corollary}[theorem]{Corollary}
\theoremstyle{definition}

\newcommand{\F}{\mathbb F}
\newcommand{\wt}{\operatorname{wt}}

\newcommand{\cG}{\mathcal G}
\newcommand{\cS}{\mathcal S}
\newcommand{\cZ}{\mathcal Z}
\newcommand{\RLG}{\operatorname{RLG}}
\newcommand{\sqcode}[4]{[\![#1,#2,#3,#4]\!]}
\newcommand{\qcode}[3]{[\![#1,#2,#3]\!]}

\newcommand{\leanfile}[1]{\texttt{\detokenize{#1}}}
\lstdefinelanguage{Lean}{morekeywords={structure,where,def,abbrev,theorem,lemma,by,exact,apply,have,show,from,fun,forall,namespace,end,variable,noncomputable,instance,match,with,if,then,else,let,in,import,extends,Prop,Type},sensitive=true,morecomment=[l]{--},morecomment=[s]{/-}{-/},morestring=[b]"}
\title{\vspace{-1.1cm}\bfseries\LARGE High-Rate Quantum Codes with Proven Distance and Low-Weight Measurements}
\author{Kishor Bharti$^{1,\star}$, Tobias Haug$^{2}$, Runzhou Tao$^{1}$, Kevin Ye$^{1}$\\[5pt]
\small $^{1}$Joint Center for Quantum Information and Computer Science (QuICS),\\[-1pt]
\small University of Maryland, College Park, Maryland 20742, USA\\[3pt]
\small $^{2}$Quantum Research Center, Technology Innovation Institute (TII), Abu Dhabi, United Arab Emirates,\\[-1pt]
\small $^\star$ Currently at IonQ}
\date{September 30, 2026}
\begin{document}
\maketitle
\begin{abstract}
We study a family of quantum subsystem codes defined on rectangular coordinate grids of any dimension, with even side lengths of at least four. The directly measured operators act along coordinate lines. For every member of the family, we derive the number of protected logical qubits and prove its dressed distance without a numerical search. To illustrate the resulting design choices, we impose a budget of at most 10,000 data qubits and determine the complete tradeoff between encoding rate and measurement weight within this family at distances 16, 32, and 64. At distance 16, one code protects 4,096 logical qubits among 10,000 data qubits using weight-ten measurements. At distances 32 and 64, codes with 7,776 and 9,216 data qubits protect 1,024 and 256 logical qubits, respectively, using weight-six measurements. We include and explain Lean code snippets throughout the paper to guide readers through the formal verification of our subsystem-code construction and its dressed and bare distances.
\end{abstract}

\section{Introduction}
How many logical qubits can a finite quantum memory protect if the operators measured to obtain error information must have low weight? Encoding rate, distance, and measurement weight constrain different parts of this question. Rate determines the protected capacity of a block. Distance sets the minimum support of an undetectable logical Pauli error. Measurement weight counts the data qubits touched by one measured operator. A useful finite construction should make the relation between these quantities explicit.

Subsystem codes allow syndrome information to be obtained by measuring gauge operators that need not commute with one another. Only the protected subsystem must be preserved; the gauge subsystem can change during error correction~\cite{poulin2005,bacon2006}. This distinction can replace a heavy stabilizer measurement by several lighter gauge measurements. It also changes the distance problem: multiplying a logical operator by a gauge operator can reduce its support, so the minimum must be taken over dressed logical operators.

We study this problem for coordinate-line gauge measurements on rectangular grids. For every number of axes $L\ge1$ and every even tuple $q_1,\ldots,q_L\ge4$, we prove
\begin{align}
 N&=\prod_{i=1}^Lq_i,&K&=\prod_{i=1}^L(q_i-2),\label{eq:mainNK}\\
 g&=N+K-2\prod_{i=1}^L(q_i-1),&D&=2^L.\label{eq:maingD}
\end{align}
Here $N$ is the number of data qubits, $K$ the number of protected logical qubits, $g$ the number of gauge qubits, and $D$ the dressed distance. The largest measured gauge has weight $w=\max_iq_i$. The distance depends only on the number of axes, whereas the rate and measurement weight depend on their lengths. This separates the choice of distance from the remaining finite optimization problem.

The construction extends the coordinate-line description of subsystem many-hypercube codes. Goto introduced many-hypercube codes based on concatenated even-length error-detecting codes~\cite{goto2024}. Nakai and Goto subsequently developed their subsystem versions, including the uniform side-four and side-six families~\cite{nakaigoto2026}. Their follow-up studies neural decoding and gauge-measurement order under circuit-level noise for selected uniform side-four codes~\cite{nakaigoto2026neural}. Mixed-side ordinary many-hypercube concatenations were also studied in Ref.~\cite{goto2025optimized}. We treat all even rectangular tuples in a common algebraic description, prove the dressed distance by a quotient-space argument, determine the minimum nonidentity stabilizer weight, and solve the resulting finite rate--weight optimization.

The matching distance bounds are constructive. For the upper bound, choose two coordinate values on every axis and apply $X$ to their $2^L$ Cartesian-product points. For the lower bound, take an arbitrary Pauli commuting with all stabilizers and quotient out the line gauges. For a nonzero logical class, a suitable binary labeling of the coordinates forces all $2^L$ labels to occur in its support. Thus no choice of gauge representative can reduce its weight below $2^L$. The same corner operator commutes with every gauge, so the bare distance is also $2^L$.

As an application, impose a budget of $10{,}000$ data qubits. At distances $16$, $32$, and $64$, the maximum-rate tuples are respectively $10^4$, $6^5$, and $4^4 6^2$, where exponents denote repeated side lengths. Their parameters are
\begin{equation}
\sqcode{10000}{4096}{974}{16},\qquad
\sqcode{7776}{1024}{2550}{32},\qquad
\sqcode{9216}{256}{5422}{64}.
\label{eq:flagships}
\end{equation}
The complete frontier includes the lighter-measurement alternatives as well. The budget is an illustrative finite constraint, not a restriction on the construction. These optimality statements concern this family and the objectives of protected rate and maximum gauge weight.

We connect these arguments to the supplied Lean formalization. Lean is an interactive theorem prover: a theorem declaration states a mathematical proposition, and its proof supplies a term that the prover checks against that proposition~\cite{lean2021}. The accompanying subsystem modules build on Lean-QEC and mathlib~\cite{ehatamm2026endtoend,mathlib2020}. They define the Pauli and binary descriptions of the code and prove the dressed and bare distance for the constructed subsystem record. We explain selected declarations beside the corresponding mathematics, including the meaning of their assumptions. The full-center equality, rank formulas, minimum stabilizer weight, and finite optimization are proved here in ordinary mathematical form; they are not claimed as formalized by the supplied subsystem modules.

Low gauge weight has a cost. We prove that every nonidentity stabilizer has weight at least $2^{L-1}\min_iq_i$. Each qubit participates in $2L$ gauges in the full-line presentation, and the coordinate grid need not be geometrically local in a physical device. Consequently, the growing-distance sequence is not a conventional quantum low-density parity-check (QLDPC) stabilizer family. Asymptotically good QLDPC constructions achieve constant rate and linear distance with bounded check weights and degrees~\cite{panteleev2022,leverrier2022}. Our result instead gives exact finite parameters and a family-specific optimization. The ideal measurement counts below do not establish fault-tolerant circuit distance or a logical error rate.

\section{Subsystem codes and their Lean representation}
\label{sec:prelim}
\subsection{Gauge equivalence and dressed distance}
We use the binary symplectic description of Pauli operators~\cite{gottesman1997,poulin2005}. Ignoring scalar phases, a Pauli on $N$ qubits is specified by $(x,z)\in\F_2^N\oplus\F_2^N$, representing $X(x)Z(z)$. Its weight is $|\operatorname{supp}(x)\cup\operatorname{supp}(z)|$. Commutation is determined by
\begin{equation}
 \langle(x,z),(x',z')\rangle=x\cdot z'+z\cdot x'.
\end{equation}
All vector-space orthogonal complements below use the binary dot product.

A subsystem code specifies a gauge group $\cG$ and a stabilizer group $\cS$. In binary notation, the stabilizer is the radical of the gauge space, namely the gauge vectors commuting with all gauge vectors. Physical stabilizer representatives are chosen with consistent phases so that $-I\notin\cS$. This phase convention matters: the literal center of a Pauli gauge group may contain scalar operators that cannot be imposed as $+1$ stabilizers. Whenever we write a center rank or support space, phases are omitted.

The stabilized Hilbert space~\cite{kribs2006} decomposes as $\mathcal H_{\rm logical}\otimes\mathcal H_{\rm gauge}$, with dimensions $2^K$ and $2^g$. If $r_G$ and $r_S$ are the binary ranks of gauge and stabilizer groups, then
\begin{equation}
 K=N-\frac{r_G+r_S}{2},\qquad
 g=\frac{r_G-r_S}{2}.
\label{eq:subsystemranks}
\end{equation}
We denote the parameters by $\sqcode{N}{K}{g}{D}$; gauge qubits are not included in the protected count $K$.

A dressed logical commutes with every stabilizer and is not a gauge operator up to phase. Its distance is
\begin{equation}
 D=\min\{\wt(P):P\in\cZ(\cS)\setminus\cG\},
\label{eq:dressed}
\end{equation}
where the set difference is understood modulo scalar phases. A bare logical additionally commutes with every gauge. Thus the dressed distance never exceeds the bare distance. These are properties of a static code; the number of faulty circuit locations needed to cause a logical error is a separate quantity.

\subsection{Reading the Lean source}
The excerpts\footnote{Lean source for code construction and proofs can be found at \url{https://github.com/VerifiedQC/Lean-QEC/tree/subsystem-codes/LeanQEC/Stabilizer/Subsystem}.} retain the variables from their surrounding source: \texttt{C} denotes a subsystem code, \texttt{r} a rectangular grid, and \texttt{l} the dimension $L$. They are excerpts from modules, rather than standalone input files. A few conventions suffice. The command \texttt{def} introduces a definition, \texttt{theorem} introduces a statement with its proof, and \texttt{structure} collects data together with properties that the data must satisfy. A colon means ``has type'': for example, \lstinline|q : Fin l → ℕ| specifies one natural number for each of the \(l\) coordinate directions. Here \texttt{Fin l} is the set \(\{0,\ldots,l-1\}\). The type \texttt{ZMod 2} is the binary field. A variable followed by a dot selects one of its fields or an associated construction. Names beginning with \texttt{h} conventionally denote hypotheses or proofs.

The formal subsystem record extends an existing stabilizer code with a gauge group:
\begin{lstlisting}
structure SubsystemCode (n : ℕ) extends StabCode n where
  gauge : Subgroup (@PauliGroup_group n)
  h_stab_le : stabs ≤ gauge
  h_stab_central :
    gauge ≤ Subgroup.centralizer
      (stabs : Set (PauliGroup_group n))
\end{lstlisting}
The two hypotheses express \(\mathcal S\subseteq\mathcal G\) and \(\mathcal G\subseteq C(\mathcal S)\). For the rectangular construction, the supplied proofs show that box stabilizers are products of gauge lines and commute with every gauge line. This record does not require that \(\mathcal S\) be the full center of \(\mathcal G\) modulo phases. The equality with the full center follows from the separate linear-algebra argument in this paper; it is not among the formal theorems in the supplied source.

The formalization includes scalar phases explicitly. It defines \texttt{trivialDressed} as the subgroup generated by the gauge group and the scalar Pauli phases. The dressed-logical predicate is the following complete definition from \texttt{Subsystem/Basic.lean}:
\begin{lstlisting}
def undetectable (E : PauliGroup n) : Prop :=
  E ∈ C.normalizerGroup ∧
  E ∉ C.trivialDressed ∧ PauliGroup.normalized E
\end{lstlisting}
The first condition says that \(E\) commutes with every stabilizer. The second says that it is not a gauge operator, even after a scalar phase is removed. The third selects a fixed phase representative and has no effect on weight. This is the distinction between a nontrivial logical error and an error that changes only the gauge subsystem.

The distance is the minimum weight of these errors. The definition returns \(N+1\) if there is no such error; this convention has no effect for our family because the cube construction supplies a nontrivial logical operator. The word \texttt{noncomputable}, used in several definitions, permits classical mathematical constructions. It does not stand for a missing proof or assert the existence of an efficient search procedure.

\subsection{The equal-sector CSS case}
The present code is of Calderbank--Shor--Steane (CSS) type: its gauge group is generated by pure $X$ and pure $Z$ operators~\cite{calderbankshor1996,steane1996}. Both sectors have the same support space $A\subseteq\F_2^N$. The stabilizer support in either sector is therefore
\begin{equation}
 S_0=A\cap A^\perp.
\label{eq:S0}
\end{equation}
Writing $a=\dim A$ and $s=\dim S_0$, we obtain $r_G=2a$, $r_S=2s$, and
\begin{equation}
 K=N-a-s,\qquad g=a-s.
\label{eq:CSSranks}
\end{equation}
Pure $X$ dressed logical classes form $S_0^\perp/A$, while pure $X$ bare logical classes form $A^\perp/S_0$. The $Z$ sectors are identical. Keeping these quotients distinct is essential: a lower bound proved only for bare operators need not bound dressed distance in a general subsystem code.

\section{Coordinate-line gauge measurements}
\label{sec:construction}
\subsection{Rectangular grids and line operators}
Let $L\ge1$ and let every $q_i\ge4$ be even. Place one qubit at each point of
\begin{equation}
 \Omega=[q_1]\times\cdots\times[q_L],\qquad [q]=\{0,\ldots,q-1\}.
\end{equation}
An axis-$i$ line fixes every coordinate except the $i$th. Include the full-line operators $X(\ell)$ and $Z(\ell)$ in the gauge generators for every such line $\ell$. We call the resulting code $\RLG(q_1,\ldots,q_L)$.

A line has $q_i$ sites. Two parallel lines are either equal or disjoint. Hence opposite-Pauli line operators on the same axis commute, because any nonempty overlap has even size. Lines on different axes can intersect at a single site, in which case opposite-Pauli operators anticommute. These are allowed gauge measurements; they are not a commuting stabilizer generating set. For $L=1$, the construction reduces to the detector $\qcode{q_1}{q_1-2}{2}$ stabilized by $X^{\otimes q_1}$ and $Z^{\otimes q_1}$.

\subsection{The grid in Lean}
The following complete record from \texttt{RLG.lean} stores the side lengths and their hypotheses.
\begin{lstlisting}
structure RLG (l : ℕ) where
  q : Fin l → ℕ
  h_ge4 : ∀ i, 4 ≤ q i
  h_even : ∀ i, Even (q i)
\end{lstlisting}
Thus an object \texttt{r : RLG l} contains both the rectangular grid and proofs that every side is even and at least four. These conditions are available wherever the object is used. They cannot be dropped accidentally when a later theorem is applied.

A coordinate is a function assigning a permitted value to each axis:
\begin{lstlisting}
def n : ℕ := ∏ i, r.q i
abbrev Coord := (i : Fin l) → Fin (r.q i)
def coordEquiv : r.Coord ≃ Fin r.n := finPiFinEquiv
\end{lstlisting}
The last line supplies a bijection between grid points and the flat labels \(0,\ldots,N-1\). This relabeling allows a geometric proof about lines and boxes to be used with the Pauli library, which indexes qubits by \texttt{Fin r.n}. It preserves supports and Hamming weights.

\subsection{Boxes of parallel lines}
For an axis $i$, choose two distinct values on each other axis and include every value on axis $i$. The resulting support is a box of $2^{L-1}$ parallel lines. Its indicator vector has the form
\begin{equation}
 h_i\otimes\bigotimes_{j\ne i}(e_{a_j}+e_{b_j}),\qquad a_j\ne b_j,
\label{eq:box}
\end{equation}
with factors placed in their coordinate order and $h_i=(1,\ldots,1)$. A same-axis line either misses this box or has even overlap $q_i$. A transverse line either misses it or meets it at two sites. Consequently, the $X$ and $Z$ operators supported on such boxes commute with every line gauge. They are also products of gauge lines. We will prove that these boxes span the full stabilizer support, rather than only a commuting subgroup of the gauge group.

\subsection{Line--box commutation in Lean}
The geometric commutation argument becomes a statement over \(\mathbb F_2\). A line support and a box support have zero binary dot product:
\begin{lstlisting}
theorem line_dot_box (i i' : Fin l) (c a b : r.Coord) :
    r.lineVec i c ⬝ᵥ r.boxVec i' a b = 0 := by
  unfold lineVec boxVec
  split_ifs with hnd
  · by_cases h : i = i'
    · subst h
      apply r.axisBox_dot_eq_zero _ _ i
      rw [lineSet_self, boxSet_self, inter_self,
        card_univ, Fintype.card_fin]
      exact r.h_even i
    · apply r.axisBox_dot_eq_zero _ _ i
      rw [lineSet_self, r.boxSet_ne _ _ h, univ_inter,
        card_pair (hnd i h)]
      exact even_two
  · simp
\end{lstlisting}
Here $\mathbin{\cdot}_v$ denotes the dot product modulo two. The proof has the same two cases as the geometric argument. If the line and box have the same direction, their overlap in that direction has even size \(q_i\); the final command uses \texttt{r.h\_even i}. If their directions differ, the relevant overlap contains the two transverse values defining the box. A degenerate box was defined to have zero support, which accounts for the final case. The commands \texttt{rw} replace an expression by an equal one, \texttt{apply} reduces a goal to the hypotheses of a known result, and \texttt{exact} supplies the required proof.

The results \texttt{stabs\_le\_gauge} and \texttt{gauge\_le\_centralizer} then assemble the subsystem record. Commutation is therefore proved for every allowed grid, without generating a large check matrix for each block size.

\subsection{The \texorpdfstring{$4\times4$}{4-by-4} example}
\label{sec:small}
On a $4\times4$ grid there are four rows and four columns, each measured in both Pauli sectors. The eight line supports per sector have one dependency: the sum of all rows equals the sum of all columns. Thus the gauge rank per sector is seven. The stabilizer rank per sector is five, as the general calculation below will show. Equations~\eqref{eq:CSSranks} then give
\begin{equation}
 \RLG(4,4)=\sqcode{16}{4}{2}{4}.
\end{equation}
The product of two parallel rows is a weight-eight stabilizer. An $X$ operator on a $2\times2$ corner has weight four and commutes with every gauge line, but is not a product of line gauges. It is a bare and dressed logical. The lower-bound proof excludes all lighter dressed logicals.

\begin{figure}[tbp]
\centering
\includegraphics[width=\linewidth]{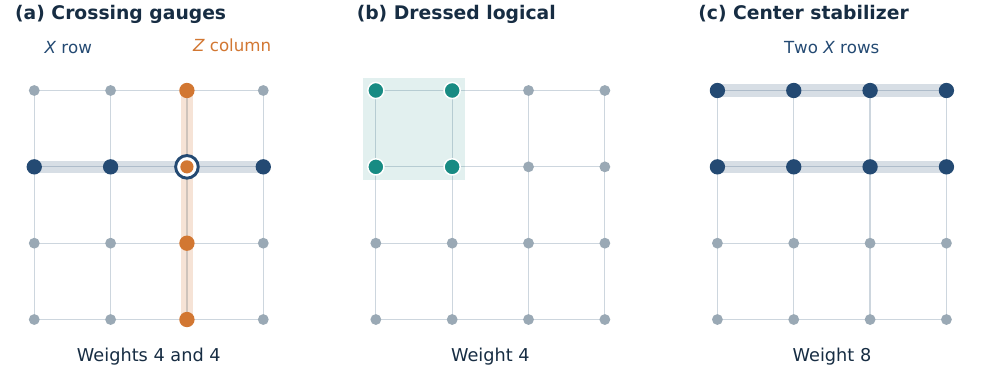}
\caption{Three supports in the $4\times4$ code. (a) A row and a column meet once, so the indicated opposite-Pauli gauge operators anticommute. Each has weight four. (b) The four highlighted sites support a minimum logical operator. (c) Two parallel rows support a weight-eight stabilizer. The colors specify supports; panels (b) and (c) can be used in either pure Pauli sector.}
\label{fig:grid}
\end{figure}

\section{Gauge rank, stabilizer rank, and protected dimension}
\label{sec:ranks}
The coordinate construction becomes a tensor-product calculation over $\F_2$. Set
\begin{equation}
 V_i=\F_2^{q_i},\quad B_i=\langle h_i\rangle,\quad C_i=h_i^\perp,
 \qquad V=\bigotimes_{i=1}^LV_i.
\end{equation}
Here $C_i$ is the even-parity code of dimension $q_i-1$ and distance two. Evenness of $q_i$ gives $B_i\subseteq C_i$, and therefore $\dim(C_i/B_i)=q_i-2$. Tensor basis vectors label physical grid sites. The gauge support is
\begin{equation}
 A=\sum_{i=1}^L V_1\otimes\cdots\otimes B_i\otimes\cdots\otimes V_L.
\label{eq:Adef}
\end{equation}
Each summand is the span of all lines parallel to one axis.

\begin{theorem}[Ranks and encoded subsystems]
\label{thm:parameters}
For every even tuple $q_i\ge4$, the gauge and stabilizer ranks per Pauli sector are
\begin{align}
 a&=N-\prod_i(q_i-1),\label{eq:arank}\\
 s&=\prod_i(q_i-1)-\prod_i(q_i-2).\label{eq:srank}
\end{align}
The full stabilizer support is spanned by the boxes in \eqref{eq:box} and equals
\begin{equation}
 S_0=\sum_{i=1}^L C_1\otimes\cdots\otimes B_i\otimes\cdots\otimes C_L.
\label{eq:center}
\end{equation}
The code protects $K=\prod_i(q_i-2)$ logical qubits and has
$g=N+K-2\prod_i(q_i-1)$ gauge qubits.
\end{theorem}

\begin{proof}
Let $\pi_i:V_i\to V_i/B_i$ be the quotient map. The kernel of $\pi=\bigotimes_i\pi_i$ is $A$. One way to see this is to choose a complement to $B_i$ in each $V_i$: in the resulting tensor decomposition, $\pi$ kills precisely the summands containing at least one $B_i$ factor. Its image has dimension $\prod_i(q_i-1)$, proving \eqref{eq:arank}.

The orthogonal complement of the line span is
\begin{equation}
 A^\perp=\bigotimes_iC_i.
\label{eq:Aperp}
\end{equation}
Indeed, orthogonality to every axis-$i$ line requires even parity on every such line. Intersecting these conditions for all axes gives the tensor product of the local even-parity spaces.

Restrict $\pi$ to $\bigotimes_iC_i$. Its image is $\bigotimes_i(C_i/B_i)$ and its kernel is the sum on the right of \eqref{eq:center}. The same kernel is $A\cap A^\perp=S_0$. This proves both the stabilizer formula and
\begin{equation}
 A^\perp/S_0\simeq\bigotimes_i(C_i/B_i).
\label{eq:barequot}
\end{equation}
Every $C_i$ is spanned by vectors $e_a+e_b$ with $a\ne b$. Hence the summands in \eqref{eq:center} are spanned by the box supports of \eqref{eq:box}; no additional stabilizer directions are missing. Subtracting the image dimension gives \eqref{eq:srank}. Finally, substitution into \eqref{eq:CSSranks} yields the stated $K$ and $g$.
\end{proof}

The proof also connects the two presentations used in this paper. The mathematical definition takes the full center of the line-gauge space. The Lean construction starts from box stabilizers and line gauges. Equation~\eqref{eq:center} identifies their binary support spaces. The formal coverage of these statements is recorded in \cref{sec:formal}.

There are $\sum_iN/q_i$ raw line measurements per sector. This is generally larger than $a$: measuring every line retains algebraic redundancy. The Lean indexing of generators may repeat the same line for different anchor points, so the number of formal indices is not a hardware measurement count.

\section{Dressed and bare logical operators}
\label{sec:distance}

The rank calculation counts protected logical qubits, but it does not
determine their distance. A dressed logical may be multiplied by any
gauge operator without changing its action on the protected subsystem.
We must therefore bound the weight of every representative of a
nonzero logical class, including representatives that do not commute
with the individual gauge generators.

\begin{lemma}[Weight in the dressed quotient]\label{lem:quotient}
Let $L\ge1$ and let $q_i\ge4$ be even. Set
$V_i=\F_2^{q_i}$, $B_i=\langle h_i\rangle$, and $C_i=h_i^\perp$,
where $h_i$ is the all-ones vector, and define
\[
 A=\sum_{i=1}^{L}
 V_1\otimes\cdots\otimes B_i\otimes\cdots\otimes V_L.
\]
Every vector representing a nonzero class in
\[
 \frac{A+\bigotimes_i C_i}{A}
 \cong\bigotimes_i(C_i/B_i)
\]
has Hamming weight at least $2^L$. The bound is attained by
$\bigotimes_i(e_0+e_1)$.
\end{lemma}

The linear functionals in the proof divide physical sites into $2^L$ classes; nontrivial logical action forces every class to contain an occupied site.

\begin{proof}
Write $W_i=V_i/B_i$, let $\pi_i:V_i\to W_i$ be the quotient map,
and set $Q_i=C_i/B_i$. Since $q_i$ is even, the parity functional
$v\mapsto\sum_a v_a$ vanishes on $B_i$ and descends to a functional
$\phi_i:W_i\to\F_2$. It satisfies
\[
 Q_i=\ker\phi_i,
 \qquad \phi_i\bigl(\pi_i(e_a)\bigr)=1
 \quad\text{for every coordinate }a.
\]
The map $\pi=\bigotimes_i\pi_i$ has kernel $A$. Moreover,
$\pi(\bigotimes_i C_i)=\bigotimes_i Q_i$, so a vector $x$
represents a nonzero class in the displayed quotient exactly when
\[
 0\ne z=\pi(x)\in\bigotimes_i Q_i.
\]

Choose linear functionals $f_i:Q_i\to\F_2$ such that
$(\bigotimes_i f_i)(z)=1$. Such a choice exists: expand $z$ in
tensor-product bases, choose a nonzero coefficient, and use the
coordinate functional for that basis element in each factor.
Extend each $f_i$ to a functional $\widetilde f_i:W_i\to\F_2$.

These functionals assign an $L$-bit label to each physical site
$a=(a_1,\ldots,a_L)$:
\[
 \gamma(a)=\bigl(
 \widetilde f_1(\pi_1(e_{a_1})),\ldots,
 \widetilde f_L(\pi_L(e_{a_L}))\bigr).
\]
We show that the support of $x$ contains a site with each of the
$2^L$ possible labels. For $b\in\{0,1\}^L$, define
\[
 F_b=\bigotimes_{i=1}^{L} g_{i,b_i},
 \qquad
 g_{i,1}=\widetilde f_i,
 \qquad
 g_{i,0}=\phi_i+\widetilde f_i.
\]
Both $g_{i,0}$ and $g_{i,1}$ restrict to $f_i$ on $Q_i$.
Consequently $F_b(z)=1$ for every $b$. On a projected physical
basis vector, however, $F_b$ equals one precisely when the site
has label $b$: the factor with $b_i=0$ takes the value
$1+\widetilde f_i(\pi_i(e_{a_i}))$.
Writing $x$ as the sum of its occupied basis vectors gives
\[
 1=F_b(z)
  =\bigl|\{a\in\operatorname{supp}(x):\gamma(a)=b\}\bigr|
       \pmod 2.
\]
Each label therefore occurs an odd, and hence positive, number
of times. The $2^L$ label classes are disjoint, proving
$\wt(x)\ge2^L$.

For the matching upper bound, each $e_0+e_1$ lies in
$C_i\setminus B_i$ because $q_i\ge4$. Its class in $Q_i$ is
nonzero, and the tensor product of these nonzero classes is
nonzero. Thus $\bigotimes_i(e_0+e_1)$ is a nontrivial class
of weight $2^L$.
\end{proof}

The lemma gives the minimum over all nonzero classes; individual
classes can have larger minimum weight. Its use of $\pi(x)$ is
essential: adding a gauge word changes $x$ but leaves this label
unchanged, so the lower bound already includes every gauge dressing.

\begin{theorem}[Exact dressed distance]\label{thm:distance}
The rectangular coordinate-line subsystem code with even side
lengths $q_i\ge4$ has
\[
 D_X=D_Z=D=2^L.
\]
The same distance is obtained when mixed $X$--$Z$ Pauli operators
are allowed.
\end{theorem}

\begin{proof}
The gauge support is $A$ in both Pauli sectors, and the center
support is $S_0=A\cap A^\perp$. An $X$-type operator commutes
with every center stabilizer precisely when its support lies in
\[
 S_0^\perp=A+A^\perp=A+\bigotimes_i C_i.
\]
Its action on the protected subsystem is nontrivial precisely
when its support does not lie in $A$. The dressed $X$-logical
classes therefore form the quotient in Lemma~\ref{lem:quotient}.
That lemma proves $D_X=2^L$. Interchanging $X$ and $Z$ gives
$D_Z=2^L$.

A mixed Pauli has binary components $(x,z)$ and physical support
$\operatorname{supp}(x)\cup\operatorname{supp}(z)$. To commute
with the center, both components must belong to $S_0^\perp$.
To act nontrivially on the protected subsystem, at least one
component must lie outside $A$. That component has weight at
least $2^L$, and its support is contained in the mixed support.
The pure-sector corner operator supplies equality for the full
Pauli distance.
\end{proof}

\begin{corollary}[Bare distance]
The minimum weight of a nontrivial bare logical operator is also
$2^L$.
\end{corollary}

\begin{proof}
A bare logical commutes with every gauge operator. In one Pauli
sector its support lies in $A^\perp$, and it is nontrivial when
it lies outside $S_0=A\cap A^\perp$. It is consequently also a
nontrivial dressed logical, so its weight is at least $2^L$.
The corner tensor $\bigotimes_i(e_0+e_1)$ belongs to
$\bigotimes_i C_i=A^\perp$ and lies outside $A$ by the lemma.
It is therefore a bare logical of weight $2^L$. The same union
of supports argument treats mixed Pauli operators.
\end{proof}

\subsection{The distance lower bound in Lean}\label{sec:lean-slicing}
The Lean source proves the lower bound by induction on the number of coordinates. We describe that argument because it explains why the distance doubles with each added coordinate and why the even-side hypothesis matters.

Let \(W_L\) be the binary span of the coordinate lines and let \(S_L\) be the span of the box supports. For \(z\in S_L^\perp\setminus W_L\), split the array along the first coordinate into slices \(z_u\), each on an \((L-1)\)-dimensional grid. Let \(W'\) denote the line span on this smaller grid. Two slices are called equivalent if \(z_u+z_v\in W'\). This is ordinary gauge equivalence for their difference.

Two consequences of the box constraints control the induction. First, every difference \(z_u+z_v\) commutes with the smaller box checks. Second,
\begin{equation}
 \sum_u z_u\in W'.
\end{equation}
To obtain the second statement, consider the boxes whose full side lies in the slicing direction. Their constraints say that the slice sum has even parity on every set containing two chosen values in each remaining coordinate. The auxiliary theorem \texttt{Grid.mem\_lineSpan\_of\_cubeSum} proves that any array with all these parities zero belongs to \(W'\).

If all slices were equivalent, choose a reference slice \(z_{u_0}\). Each difference \(z_u+z_{u_0}\) could be generated by lines within slice \(u\), while the repeated array \(z_{u_0}\) could be generated by lines in the first direction. These two contributions reconstruct \(z\), contradicting \(z\notin W_L\). There is consequently an inequivalent pair. For every such pair, the induction hypothesis and the Hamming-weight triangle inequality give
\begin{equation}
 \operatorname{wt}(z_u)+\operatorname{wt}(z_v)
 \geq \operatorname{wt}(z_u+z_v)
 \geq 2^{L-1}.
\end{equation}

It remains to find two disjoint inequivalent pairs. Pass to the quotient by \(W'\), and write \(a_u=[z_u]\). The cosets are not all equal and satisfy \(\sum_u a_u=0\). Since the number of slices is even, they cannot consist of one exceptional coset and an otherwise constant list: the sum would be the difference between the exceptional and common cosets. With at least four slices, every other nonconstant list has two disjoint unequal pairs. Indeed, start with one unequal pair. If the remaining entries contain an unequal pair, we are done. Otherwise all remaining entries equal one coset; either both selected entries differ from that coset and can be paired with two remaining entries, or there is just one exceptional entry. The latter case was excluded. Summing the two pair bounds gives \(\operatorname{wt}(z)\geq2^L\).

The formal statement below is reproduced from \texttt{GridLines.lean}; its proof body is omitted from this excerpt.
\begin{lstlisting}
theorem two_pow_le_hammingNorm :
    ∀ (l : ℕ) (α : Fin l → Type u)
      [∀ i, Fintype (α i)] [∀ i, DecidableEq (α i)],
    (∀ i, Even (Fintype.card (α i)) ∧
      4 ≤ Fintype.card (α i)) →
    ∀ z : Vec α, Cent z → z ∉ lineSpan α →
      2 ^ l ≤ hammingNorm z
\end{lstlisting}
Here \texttt{Fintype} specifies finite coordinate sets, \texttt{Cent z} is the box-commutation condition, and \texttt{hammingNorm z} is its support size. The statement is independent of the side lengths apart from their parity and lower bound. The formal proof starts at dimension zero, where a nonzero binary vector on the single grid point has weight one.

\subsection{Combining the two distance bounds}
The cube support contains two chosen values in each coordinate and therefore has \(2^L\) points. It meets every coordinate line in zero or two points, so its \(X\) operator commutes with every gauge operator. To show that it is not a product of gauge lines, choose a second cube sharing exactly one corner with it. Every gauge line has even overlap with the second cube, whereas the first cube has odd overlap. This supplies the required witness outside the gauge group.

The final dressed-distance proof in \texttt{RLGDistance.lean} is short because the witness and lower bound have already been established:
\begin{lstlisting}
theorem distance_eq :
    r.toSubsystemCode.distance = 2 ^ l := by
  obtain ⟨a, b, hab⟩ := r.exists_cube
  exact SubsystemCode.distance_eq_of _
    (r.cube_undetectable hab) (r.pauli_weight_cube hab)
    (fun F hF => r.two_pow_le_pauli_weight F hF)
\end{lstlisting}
The first line of the proof chooses a nondegenerate cube. The underscore asks Lean to infer an argument; \texttt{fun F hF} introduces an arbitrary Pauli \(F\) and a hypothesis \texttt{hF} that it is a dressed logical. The next lines supply three facts to \texttt{distance\_eq\_of}: this cube is a nontrivial dressed logical, its weight is \(2^L\), and every dressed logical has at least that weight. The proof applies to mixed \(X\), \(Y\), and \(Z\) errors as well: at least one binary Pauli component lies outside the line span, and its Hamming weight is no larger than the Pauli support.

The same cube commutes with the entire gauge group and hence is also a bare logical. Since every bare logical is a dressed logical, the same lower bound gives bare distance \(2^L\). The source records this equality as \texttt{RLG.bareDistance\_eq}.

\section{Minimum stabilizer weight}
\label{sec:centerweight}
The distance concerns logical operators. A different minimum determines whether the code also has a low-weight commuting stabilizer presentation. In the present family, gauge operators can be much lighter than every nonidentity stabilizer.

\begin{theorem}[Minimum nonidentity stabilizer weight]
\label{thm:centerweight}
For $L\ge1$ and even side lengths $q_i\ge4$, the minimum weight of a nonzero word in the stabilizer support is
\begin{equation}
 d(S_0)=2^{L-1}\min_iq_i.
\label{eq:centerdistance}
\end{equation}
This is also the minimum weight of a nonidentity Pauli stabilizer, including mixed Pauli operators.
\end{theorem}

\begin{proof}
For $L=1$, the stabilizer support is $B_1$, whose unique nonzero word has weight $q_1$. Suppose the claim holds in $L-1$ dimensions. Write $C_{<L}=\bigotimes_{i<L}C_i$ and let $S_{<L}$ be the smaller stabilizer support. Separating the last tensor factor in \eqref{eq:center} gives
\begin{equation}
 S_0=C_{<L}\otimes B_L+S_{<L}\otimes C_L.
\label{eq:centerrecursion}
\end{equation}
View a nonzero word of $S_0$ as $q_L$ slices in $C_{<L}$. Equation~\eqref{eq:centerrecursion} implies that all slices have the same coset modulo $S_{<L}$ and their sum is zero.

The tensor product $C_{<L}$ has minimum distance $2^{L-1}$. For completeness, the usual product-distance argument is inductive: a nonzero coordinate fiber of a tensor-code word occupies at least the distance of one factor, and each resulting nonzero slice is a word in the remaining tensor code. Product minimum words attain the bound.

If the common slice coset is nonzero, all $q_L$ slices are nonzero, giving weight at least $q_L2^{L-1}$. If the common coset is zero, every slice lies in $S_{<L}$, and at least two are nonzero because their sum vanishes. Induction then gives weight at least $2^{L-1}\min_{i<L}q_i$. The smaller of these two bounds is \eqref{eq:centerdistance}.

Choose a shortest axis $j$. The box
\begin{equation}
 h_j\otimes\bigotimes_{i\ne j}(e_0+e_1)
\end{equation}
has weight $q_j2^{L-1}$ and attains the bound. Finally, a mixed stabilizer has support equal to the union of its two binary components. At least one component is a nonzero word of $S_0$, so a mixed stabilizer cannot be lighter.
\end{proof}

For $\RLG(6^5)$, the measured gauges have weight six, whereas the lightest nonidentity stabilizer has weight $96$. For $\RLG(4^4 6^2)$, the corresponding weights are six and $128$. No choice of stabilizer generating basis can avoid this lower bound. The result explains precisely what the gauge freedom reduces: the supports of the directly measured operators, while the stabilizer constraints remain heavy. The supplied Lean code constructs these box stabilizers and proves their commutation; it does not include the minimum-weight theorem for their span.

\section{The finite rate--measurement-weight tradeoff}
\label{sec:finite}
For a fixed distance $D=2^L$, the rate and maximum gauge weight are
\begin{equation}
 R(q_1,\ldots,q_L)=\prod_i\left(1-\frac2{q_i}\right),
 \qquad w(q_1,\ldots,q_L)=\max_iq_i.
\label{eq:rateweight}
\end{equation}
A tuple dominates another if it has at least as large a rate and at most as large a weight, with one strict inequality. The Pareto frontier consists of tuples not dominated by another allowed tuple. We identify tuples up to axis permutations. Capacity $K$ is a different objective from rate $K/N$.

\begin{theorem}[Complete frontier at distances 16, 32, and 64]
\label{thm:finiteopt}
Among even rectangular tuples with $q_i\ge4$ and $N\le10{,}000$, the complete frontiers for maximizing $R$ and minimizing $w$ are
\begin{equation}
\begin{array}{c|l}
D&\text{side tuples, up to permutation}\\\hline
16&4^4,\ 6^4,\ 8^4,\ 10^4,\\
32&4^5,\ 6^5,\\
64&4^6,\ 4^4 6^2.
\end{array}
\label{eq:frontier}
\end{equation}
The last tuple in each row uniquely maximizes the protected rate at that distance. No member under this data-qubit budget has distance $128$ or greater.
\end{theorem}

\begin{proof}
Since $N\ge4^L$, the budget gives $L\le6$. This proves the maximum-distance statement.

For $L=4$, the function $h(y)=\log(1-2e^{-y})$ is increasing and strictly concave for $y>\log2$. Hence
\begin{equation}
 \log R\le4h\!\left(\frac14\sum_i\log q_i\right)
 \le4h(\log10).
\end{equation}
Equality requires $q_i=10$ for every axis. Under weight caps four, six, and eight, the corresponding uniform tuples are feasible and maximize each increasing factor in \eqref{eq:rateweight}. Tuples with weight above ten are dominated by $10^4$. This gives the distance-16 row.

For $L=5$, a tuple with no side four has every side at least six. Any side at least eight would force $N\ge8\cdot6^4=10{,}368$. Thus $6^5$ is the only such tuple, with rate $32/243$.

If exactly one side is four, the product of the remaining four sides is at most $2{,}500$. At least two of those sides equal six: otherwise the product would be at least $6\cdot8^3>2{,}500$. After fixing two sixes, the last sorted even pair has product below $70$. Its possibilities are $(6,6)$, $(6,8)$, $(6,10)$, and $(8,8)$. The last pair gives the largest rate, namely
\begin{equation}
 \frac12\left(\frac23\right)^2\left(\frac34\right)^2=\frac18.
\end{equation}
If at least two sides are four, fix two of them. The other three have product at most $625$. The same concavity bound gives
\begin{equation}
 R\le\frac14\left(1-\frac2{\sqrt[3]{625}}\right)^3
 <\frac14\left(\frac79\right)^3<\frac18.
\end{equation}
Since $1/8<32/243$, the unique rate maximum is $6^5$. Every weight-at-least-six alternative is dominated by it. The tuple $4^5$ has the smallest possible weight and remains on the frontier.

For $L=6$, three sides above four would imply $N\ge4^3 6^3=13{,}824$. If two sides exceed four, both must be six, because $4^4\cdot6\cdot8>10{,}000$. If exactly one side exceeds four, it is six or eight. The complete list is therefore
\begin{equation}
 4^6,\quad4^5 6,\quad4^5 8,\quad4^4 6^2,
\end{equation}
with rates $1/64$, $1/48$, $3/128$, and $1/36$. The last tuple dominates the middle two, and $4^6$ has smaller measured weight. This proves the remaining row.
\end{proof}

\begin{table}[tbp]
\centering
\caption{The complete frontiers under $N\le10{,}000$. The column $w$ is the largest measured gauge weight; $d(S_0)$ is the minimum nonidentity stabilizer weight. Exponents denote repeated side lengths.}
\label{tab:frontier}
\small
\begin{tabular}{@{}lrrrrrrr@{}}
\toprule
Sides&$N$&$K$&$g$&$D$&$K/N$&$w$&$d(S_0)$\\
\midrule
$4^4$&256&16&110&16&$1/16$&4&32\\
$6^4$&1,296&256&302&16&$16/81$&6&48\\
$8^4$&4,096&1,296&590&16&$81/256$&8&64\\
$10^4$&10,000&4,096&974&16&$256/625$&10&80\\
\midrule
$4^5$&1,024&32&570&32&$1/32$&4&64\\
$6^5$&7,776&1,024&2,550&32&$32/243$&6&96\\
\midrule
$4^6$&4,096&64&2,702&64&$1/64$&4&128\\
$4^4 6^2$&9,216&256&5,422&64&$1/36$&6&128\\
\bottomrule
\end{tabular}
\end{table}

\begin{figure}[tbp]
\centering
\includegraphics[width=.9\linewidth]{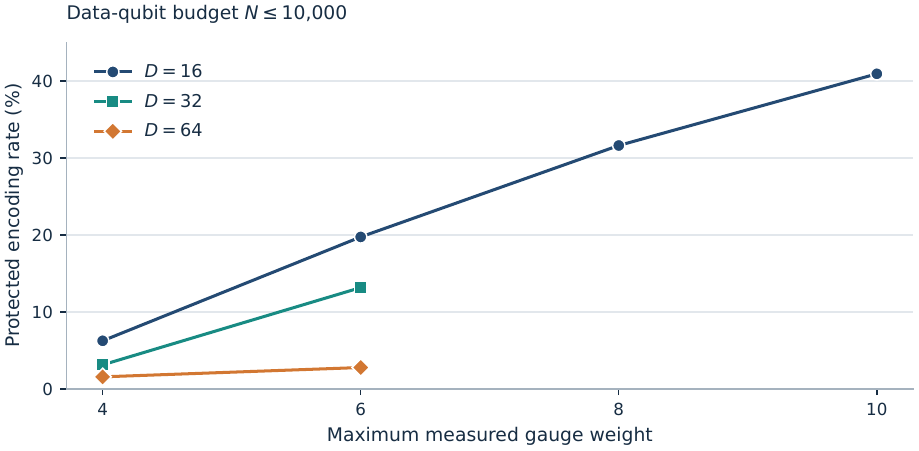}
\caption{Exact rate--weight frontiers at fixed distance under $10{,}000$ data qubits. The points are the discrete choices in \cref{tab:frontier}; line segments guide the eye and do not represent intermediate codes. Increasing the number of axes raises distance and reduces the maximum achievable rate under this budget.}
\label{fig:frontier}
\end{figure}

For a strict measured-weight cap $w<10$, the maximum-rate distance-16 member becomes $\RLG(8^4)$; the distance-32 and distance-64 optima are unchanged. At distance 32, the tuple $(4,6,6,8,8)$ protects $1{,}152$ qubits, more than the $1{,}024$ protected by $6^5$, but has lower rate $1/8$ and larger measured weight eight. It is therefore absent from the rate--weight frontier. This example shows why the objective must be specified when claiming an optimum.

These finite results use the distance theorem rather than a distance search. An independent exact-arithmetic enumeration gives $514$, $55$, and $4$ sorted tuples for $L=4,5,6$, respectively, and reproduces \eqref{eq:frontier}. The independent verification program also reconstructs small binary gauge matrices and checks their ranks and witnesses. The enumeration procedure appears in Appendix~\ref{app:verification}.

\section{Measurement counts and ideal schedules}
\label{sec:resources}
A line-gauge presentation specifies both how many operators are measured and how often each data qubit participates. These counts are distinct from the number of independent gauge generators.

\begin{proposition}[Full-line measurement counts]
\label{prop:resources}
For $\RLG(q_1,\ldots,q_L)$, the numbers of distinct full-line operators are
\begin{align}
 M_X=M_Z&=\sum_i\frac{N}{q_i},&M&=2\sum_i\frac{N}{q_i},\label{eq:measurements}\\
 E&=2LN,&\Delta_X=\Delta_Z&=L.\label{eq:incidences}
\end{align}
Here $E$ counts data--gauge incidences and $\Delta_X+\Delta_Z=2L$ is the total data degree. There are $M_X-a$ binary row dependencies per sector.

Consider an ideal circuit using one reusable ancilla per active line, with all lines of one axis measured in parallel. Measure every $X$-axis batch before the $Z$-axis batches. The schedule uses
\begin{equation}
 A_{\max}=\frac{N}{\min_iq_i},\qquad
 2L\ \text{batches},\qquad
 2\sum_iq_i\ \text{two-qubit gate layers}.
\label{eq:schedule}
\end{equation}
The layer count assumes direct ancilla--data interactions and excludes preparation, readout, reset, routing, and protection against fault propagation.
\end{proposition}

\begin{proof}
The lines parallel to axis $i$ partition the grid into $N/q_i$ disjoint supports. Summing over axes and sectors proves \eqref{eq:measurements}. Each qubit belongs to one line per axis in each sector, proving \eqref{eq:incidences}. Subtracting the rank $a$ gives the dependency count.

Within an axis batch, disjoint lines can use separate ancillas without gate collisions. Each line ancilla interacts with its $q_i$ data qubits in $q_i$ successive layers. Reusing the same ancilla pool across batches gives \eqref{eq:schedule}. All $X$ gauges commute with one another, as do all $Z$ gauges, so this sector ordering allows products of same-Pauli outcomes to infer their stabilizer products in the ideal circuit.
\end{proof}

\begin{table}[tbp]
\centering
\caption{Counts for the three maximum-rate examples. Dependencies include both Pauli sectors. ``Data + ancillas'' uses the reusable pool in \eqref{eq:schedule}; it is not a fault-tolerant footprint.}
\label{tab:resources}
\small
\begin{tabular}{@{}lrrr@{}}
\toprule
Quantity&$10^4$&$6^5$&$4^4 6^2$\\
\midrule
Distinct gauge measurements $M$&8,000&12,960&24,576\\
Gauge rank $2a$&6,878&9,302&14,382\\
Stabilizer rank $2s$&4,930&4,202&3,538\\
Row dependencies $M-2a$&1,122&3,658&10,194\\
Data--gauge incidences $E$&80,000&77,760&110,592\\
Total data degree $2L$&8&10&12\\
Reusable ancillas $A_{\max}$&1,000&1,296&2,304\\
Data + ancillas&11,000&9,072&11,520\\
Axis/Pauli batches&8&10&12\\
Ideal two-qubit gate layers&80&60&56\\
\bottomrule
\end{tabular}
\end{table}

A data-qubit cap and a live-qubit cap yield different optimizations. Under $N+A_{\max}\le10{,}000$, exact enumeration gives maximum-rate tuples $(8,10,10,10)$, $6^5$, and $4^5 6$ at distances $16$, $32$, and $64$. Their data-plus-ancilla counts are $9{,}000$, $9{,}072$, and $7{,}680$, respectively. This calculation uses precisely the ideal ancilla policy in \eqref{eq:schedule}.

A smaller ancilla pool trades space for time. With only $512$ reusable line ancillas, $\RLG(4^4 6^2)$ fits in $9{,}728$ data-plus-ancilla qubits. Each length-four axis must be split into five subbatches, and each length-six axis into three, giving $52$ subbatches and $232$ ideal two-qubit layers across both sectors. The code parameters are unchanged.

In a noisy circuit, a single ancilla fault can propagate to several data qubits. Gauge outcomes from different times also cannot be multiplied indiscriminately to create error detectors. Flags, verified ancillas, measurement repetition, and an explicit circuit analysis may change both the footprint and time. Equations~\eqref{eq:measurements}--\eqref{eq:schedule} are exact counts for the stated ideal schedule; they are not a claim that the circuit inherits distance $2^L$.

\section{Relation to existing codes}
\label{sec:comparison}
The uniform $q=4$ and $q=6$ members are subsystem many-hypercube codes~\cite{nakaigoto2026}. Neural decoding and gauge-measurement scheduling have been studied under circuit-level noise for selected uniform side-four members~\cite{nakaigoto2026neural}; those results do not establish the performance of the optimized rectangular codes considered here. Mixed side lengths have also been considered for ordinary concatenated many-hypercube codes~\cite{goto2025optimized}. The present analysis treats the full-line gauge group for arbitrary even rectangular tuples, determines its complete stabilizer support, and proves the dressed-distance lower bound for every gauge representative. The finite optimization then selects a tuple using proved parameters rather than estimating distance separately for each candidate.

Subsystem products provide a useful comparison with a different relation between rate and distance. For two classical $[n_i,k_i,d_i]$ codes with $k_1,k_2>0$, the standard subsystem product has $N=n_1n_2$, $K=k_1k_2$, and dressed distance $\min(d_1,d_2)$~\cite{baconcasaccino2006,liyoder2020}. Applying this construction to two binary simplex codes produces the SHYPS family~\cite{malcolm2026},
\begin{equation}
 \sqcode{(2^r-1)^2}{r^2}{(2^r-1-r)^2}{2^{r-1}},
\end{equation}
with weight-three gauge generators for $r\ge3$. At $r=6$, it gives $\sqcode{3969}{36}{3249}{32}$. The uniform $6^5$ line-gauge code has the same static distance and a protected rate $32/243\simeq13.17\%$, compared with $36/3969\simeq0.91\%$ for SHYPS(6). It also uses more data qubits and doubles the gauge weight from three to six. The SHYPS work develops logical Clifford operations and circuit-level simulations, which are outside the results established here. Static parameter comparisons do not determine which code is preferable for a complete computation.

High-rate stabilizer constructions, including bivariate bicycle codes~\cite{bravyi2024}, measure commuting checks rather than the noncommuting gauges used here. Hong constructs balanced-product QLDPC codes with design rate $1/5$ and weight-nine checks, and evaluates their decoding performance under idling-free circuit-level noise~\cite{hong2026}. Their measured generators should therefore be compared with our gauges only together with a specified extraction circuit. A numerical logical witness proves an upper bound on distance. A finite exhaustive certificate proves an exact distance when accompanied by a matching witness. Such a certificate is as rigorous for its stated instance as a symbolic proof; the advantage of the present formula is that it applies to every allowed tuple without repeating a finite-instance distance computation.

For a uniform fixed side length $q$ and growing $L$,
\begin{equation}
 N=q^L,\qquad
 K=N^{\log_q(q-2)},\qquad
 D=N^{\log_q2},\qquad
 R=N^{\log_q(1-2/q)}.
\end{equation}
The rate tends to zero, the full-line data degree is $2\log_qN$, and the lightest stabilizer has weight $(q/2)D$. Thus this family does not have the asymptotic guarantees of good QLDPC stabilizer codes~\cite{panteleev2022,leverrier2022}. Its useful feature is the direct control of finite rate, exact distance, and gauge weight, with a simple description of every measured support.

\section{Formalization and reproducibility}
\label{sec:formal}
The Lean source is available in the \href{https://github.com/VerifiedQC/Lean-QEC/tree/subsystem-codes/LeanQEC/Stabilizer/Subsystem}{Lean-QEC GitHub repository}. The four subsystem modules used here, \leanfile{Basic.lean}, \leanfile{RLG.lean}, \leanfile{GridLines.lean}, and \leanfile{RLGDistance.lean}, are fixed at \href{https://github.com/VerifiedQC/Lean-QEC/tree/eccb50c389b4d5d9feab3e29809b23aa1761fa14/LeanQEC/Stabilizer/Subsystem}{commit \texttt{eccb50c389b4d5d9feab3e29809b23aa1761fa14}}. This snapshot also records the Lean toolchain and dependency revisions.

The formal proof establishes dressed and bare distance for the constructed subsystem record. The full-center equality, rank formulas, minimum stabilizer weight, measurement counts, and finite frontiers are established in this paper and are not claimed as Lean theorems.

\section{Discussion}

We prove that the dressed distance of rectangular coordinate-line subsystem codes is fixed entirely by the number of axes, despite the freedom to modify logical operators by gauge multiplication. For every grid with even side lengths of at least four, the distance is exactly $2^L$. At fixed distance, the side lengths can therefore be chosen to balance encoding rate and measurement weight using explicit formulas. The $10{,}000$-qubit frontiers demonstrate how this structure turns code selection into a finite optimization with proved parameters. The same approach applies to other qubit budgets and measurement-weight constraints.

The stabilizer analysis reveals where the benefit of gauge measurements arises. Although each measured operator acts along a single coordinate line, every nonidentity stabilizer has weight at least $2^{L-1}\min_i q_i$. Gauge freedom thus allows the information carried by increasingly heavy stabilizers to be accessed through lighter measurements. This distinction also identifies a central implementation challenge: recovering that information reliably requires control of errors in both the data and the measurement record. The minimum stabilizer weight makes clear that choosing a different stabilizer basis cannot remove this challenge.

The Lean development provides a complementary proof of the distance theorem by induction on the number of axes, with explicit treatment of Pauli phases and gauge equivalence. Extending it to the full-center equality and rank formulas would yield a complete formal verification of the code parameters. For implementation, the next step is to develop and evaluate decoders and fault-tolerant extraction circuits for the rectangular choices studied here, building on the decoding and measurement-order studies for uniform subsystem many-hypercube codes~\cite{nakaigoto2026,nakaigoto2026neural}. These must account for the timing of gauge measurements, propagation of ancilla faults, and the connectivity available in a physical device. The resulting circuit distance and logical error rates may favor different side lengths from those selected by static rate and measurement weight alone. Establishing this connection would determine when the encoding advantage translates into a reduction in the resources needed for a reliable quantum memory.

\appendix
\section[Exact finite enumeration]{Exact finite enumeration}
\label{app:verification}
For fixed $L$, every allowed tuple can be enumerated in nondecreasing order. At a recursion node, let $p$ be the product of the chosen sides and let $r$ sides remain. A candidate next side $q$ can be accepted only if $pq^r\le B$, where $B$ is the data-qubit budget. All later sides must be at least $q$, so this condition is also an exact stopping test. No floating-point root or distance search is needed.

An exact-arithmetic implementation of this recursion compares rates as rational numbers. For $B=10{,}000$, it enumerates $514$, $55$, and $4$ tuples at $L=4,5,6$. Grouping the middle case by the number of sides equal to four gives
\begin{center}
\begin{tabular}{@{}crlc@{}}
\toprule
Number of fours&Tuples&Largest-rate tuple in the group&Rate\\
\midrule
0&1&$6^5$&$32/243$\\
1&4&$(4,6,6,8,8)$&$1/8$\\
2&11&$(4,4,6,10,10)$&$8/75$\\
3&21&$(4,4,4,12,12)$&$25/288$\\
4&17&$(4,4,4,4,38)$&$9/152$\\
5&1&$4^5$&$1/32$\\
\bottomrule
\end{tabular}
\end{center}
The program also checks gauge and stabilizer rank identities, the raw measurement counts and weight histograms, both qubit-budget policies, and small explicit binary line matrices. It exhausts the dressed quotient for the $\sqcode{4}{2}{0}{2}$ and $\sqcode{16}{4}{2}{4}$ examples. These finite checks passed. They support the numerical transcription and small instances; the general distance claim rests on the proofs in the text and the source declarations described above.

\bibliographystyle{unsrtnat}
\bibliography{references}
\end{document}